\documentclass[10pt]{article}

\input{glyphtounicode}
\usepackage[utf8]{inputenc}
\usepackage[T1]{fontenc}
\usepackage{lmodern}
\usepackage[margin=0.95in]{geometry}
\usepackage{amsmath,amssymb,amsthm,mathtools}
\usepackage{graphicx}
\usepackage{booktabs}
\usepackage{tabularx}
\usepackage[expansion=false,protrusion=true]{microtype}
\usepackage{xcolor}
\usepackage{tikz}
\usetikzlibrary{positioning,arrows.meta,calc}
\usepackage[colorlinks=true,linkcolor=blue!55!black,citecolor=blue!55!black,urlcolor=blue!55!black]{hyperref}
\usepackage{orcidlink}
\usepackage[numbers,sort&compress]{natbib}
\usepackage{caption}
\usepackage{titlesec}
\titleformat*{\section}{\large\bfseries}
\titleformat*{\subsection}{\normalsize\bfseries}
\titlespacing*{\section}{0pt}{1.05ex plus .2ex minus .2ex}{0.65ex plus .1ex}
\titlespacing*{\subsection}{0pt}{0.85ex plus .2ex minus .2ex}{0.45ex plus .1ex}

\newtheorem{proposition}{Proposition}
\newtheorem{designhypothesis}{Design Hypothesis}
\newtheorem{assumption}{Assumption}
\newtheorem{remarkx}{Remark}

\newcommand{\aria}{\textsc{ARIA}}
\newcommand{\Kcoll}{\mathcal{K}}
\newcommand{\eds}{\mathrm{EDS}}
\makeatletter
\renewcommand{\@fnsymbol}[1]{\ifcase#1\or \ensuremath{\dagger}\else\@arabic{#1}\fi}
\newcommand{\authorthanks}[1]{\stepcounter{footnote}\protected@xdef\@thanks{\@thanks\protect\footnotetext[\the\c@footnote]{#1}}}
\makeatother

\title{\textbf{Compliant with Local Controls, Collectively Discriminatory}\\[4pt]
\large A Governance Architecture for Multi-Agent AI in Regulated Finance}

\author{%
\begin{tabularx}{0.96\textwidth}{@{}>{\centering\arraybackslash}X>{\centering\arraybackslash}X>{\centering\arraybackslash}X@{}}
Jose Manuel de la Chica Rodr\'iguez\,\orcidlink{0009-0009-9649-5805}\textsuperscript{\ensuremath{\dagger}} & Juan Manuel Vera D\'iaz\,\orcidlink{0000-0002-6152-5789} & Pablo Delgado Romero\,\orcidlink{0009-0001-6195-2676}\\
\small Santander AI Lab & \small Santander AI Lab & \small Santander AI Lab\\
\small Grupo Santander & \small Grupo Santander & \small Grupo Santander\\
\small Madrid, Spain & \small Madrid, Spain & \small Madrid, Spain
\end{tabularx}%
\authorthanks{Correspondence: Jose Manuel de la Chica Rodr\'iguez, josema.chica@gruposantander.com.\protect\\[1pt]The views expressed are the authors' own and do not represent positions of Grupo Santander.}%
}
\date{September 2026}

\begin{document}
\maketitle

\begin{abstract}
Financial institutions are beginning to deploy agentic workflows in credit, fraud, collections, compliance, and operational control. Governance remains largely component-centric: each model or agent is specified, tested, authorized, and monitored locally. That is insufficient when institutional risk arises from the joint behavior of many locally acceptable components. We call this gap \emph{constitutional non-compositionality}: local compliance checks need not compose into acceptable collective outcomes such as bounded disparate impact, market integrity, or traceable accountability.

We propose \aria{} as a finance-specific reference architecture and falsifiable research agenda for agent-population governance. It organizes six capabilities across normative-accountability, execution-control, and assurance-learning planes: policy specification, population-level observed-versus-expected behavior monitoring ($M_2$), bounded authority, runtime containment, adaptive policy change, and preserved human oversight competence. Two simulations illustrate shared-signal thin-file exclusion under local controls and earlier warning from observed-versus-expected distributional monitoring in a constructed drift regime. The contribution maps these controls to fair-lending, EU AI Act, model-risk, and conduct-supervision evidence needs, and closes with a validation agenda rather than a production-effectiveness claim.
\end{abstract}

\vspace{2pt}
\normalsize

\section{Introduction}\label{sec:intro}

Regulated financial institutions are beginning to deploy agentic workflows across credit, fraud, collections, compliance, and operational control. In these settings, assurance is still commonly organized around individual models or components: each component is specified, tested, authorized, and monitored against local fairness, safety, conduct, or model-risk requirements. This component-centric approach can be insufficient when institutional outcomes are produced by many interacting agents that share signals, data sources, tools, authority boundaries, or feedback loops. A credit-decisioning population, for example, may satisfy local constraints for each agent while producing an aggregate approval pattern that disadvantages thin-file applicants relative to otherwise similar thick-file applicants. Such an outcome can create disparate-impact exposure and supervisory questions under fair-lending law \cite{ecoa}, role-dependent obligations under the EU AI Act \cite{aiact}, and revised model-risk governance expectations \cite{sr262}, even when no single component violates its local control.

Our approach treats that scenario as a structural weakness in component-centric AI assurance. The weakness is not that model-level assurance is useless; it is that model-level assurance is incomplete once outcomes are produced by interacting or coupled agentic workflows. Closing the gap requires governance machinery aimed at the population and workflow, not only at the individual model.

This contribution is a finance-specific reference architecture and falsifiable industrial research agenda. It is not a production validation, a legal interpretation, an empirical estimate of real banking effects, or evidence that Santander currently deploys the systems described. \aria{} contributes a finance-specific synthesis that integrates population-level outcome constraints, revocable authority, runtime containment, governed policy change, human oversight competence, and layer-specific audit artifacts into a single research architecture.

An \emph{agentic workflow population} consists of two or more adaptive or stateful decision-making components whose actions, shared state, delegated authority, common tools, or common environmental feedback jointly affect an institutional outcome. This includes explicitly communicating multi-agent systems, orchestrated agent workflows, ensembles of decision services, and coupled model-based controls; it excludes ordinary single-model deployments unless their interaction with other components creates a joint outcome that local controls cannot assess. We use \emph{constitutional non-compositionality} as an umbrella term for governance failures in which component-level constraints, policies, permissions, or accountability assignments do not compose into institutional-level constraints.

Our contribution is threefold: (1) a six-part taxonomy of non-compositional governance failures; (2) a three-plane architecture of six auditable governance capabilities for finance; and (3) a falsifiable validation agenda with metrics, baselines, costs, and accountable artifacts.

The argument proceeds in four steps:
\begin{itemize}
\item \textbf{First}, we define constitutional non-compositionality as the failure of local-control guarantees to compose into collective compliance guarantees (Section~\ref{sec:noncomp}). The claim is modest: individual checks do not by themselves imply acceptable population-level outcomes. Empirical work on large-language-model (LLM) collusion, covert coordination, emergent misalignment, Institutional AI, and financial multi-agent fairness shows why the concern is not merely hypothetical \cite{fish2024collusion,motwani2024secret,madigan2026emergentbias,ranjan2025fairnessagentic}.
\item \textbf{Second}, the blind spot is a compliance exposure, not just a technical defect. In regulated finance, the duties that matter often attach to institutional outcomes. Disparate-impact doctrine examines lending patterns, not only component intent. EU AI Act obligations around risk management, oversight, deployer controls, and provider-side post-market monitoring where applicable create evidence needs about system behavior \cite{aiact}. SR~11-7 supplied the historical template for model-risk governance; SR~26-2 is now the revised U.S. supervisory reference point \cite{sr117,sr262}. Financial Conduct Authority (FCA) Consumer Duty guidance, Prudential Regulation Authority (PRA) model-risk expectations, Monetary Authority of Singapore (MAS) Fairness, Ethics, Accountability and Transparency (FEAT) principles \cite{fca2022consumerduty,pra2023ss123,feat}, Singapore Infocomm Media Development Authority (IMDA) guidance \cite{imda2026}, and the U.S. National Institute of Standards and Technology (NIST) initiative \cite{nist2026} are consistent with a move toward observability, traceability, bounded autonomy, and intervention rights. \aria{} is offered as one way to operationalize those evidence needs, not as a statement that law already mandates this architecture.
\item \textbf{Third}, Section~\ref{sec:architecture} introduces \aria{} as six governance capabilities organized across three architectural planes: a normative and accountability plane, an execution-control plane, and an assurance and learning plane. The architecture separates capability from authority, links population monitors to containment, and names the artifacts associated with internal audit, model risk, compliance, business ownership, and supervisory evidence.
\item \textbf{Fourth}, Section~\ref{sec:experiments} reports two mechanism demonstrations. The first illustrates shared-signal and measurement-quality coupling in a credit setting: a local fairness rule conditional on each agent's own estimate can coexist with aggregate disparity. The second illustrates that, in one deliberately constructed two-stage drift regime, an observed-versus-expected distributional monitor can alarm earlier than a prohibited-outcome monitor. The contribution then turns to operationalization, field validation, the standardization path, related work, and limitations that specify what would falsify or narrow the proposal.
\end{itemize}

Figure~\ref{fig:overallaria} summarizes this roadmap by contrasting component-centric local assurance with the population-and-workflow governance layer that \aria{} adds before the regulatory and technical motivation.

\begin{figure}[t]
\centering
\includegraphics[width=\textwidth]{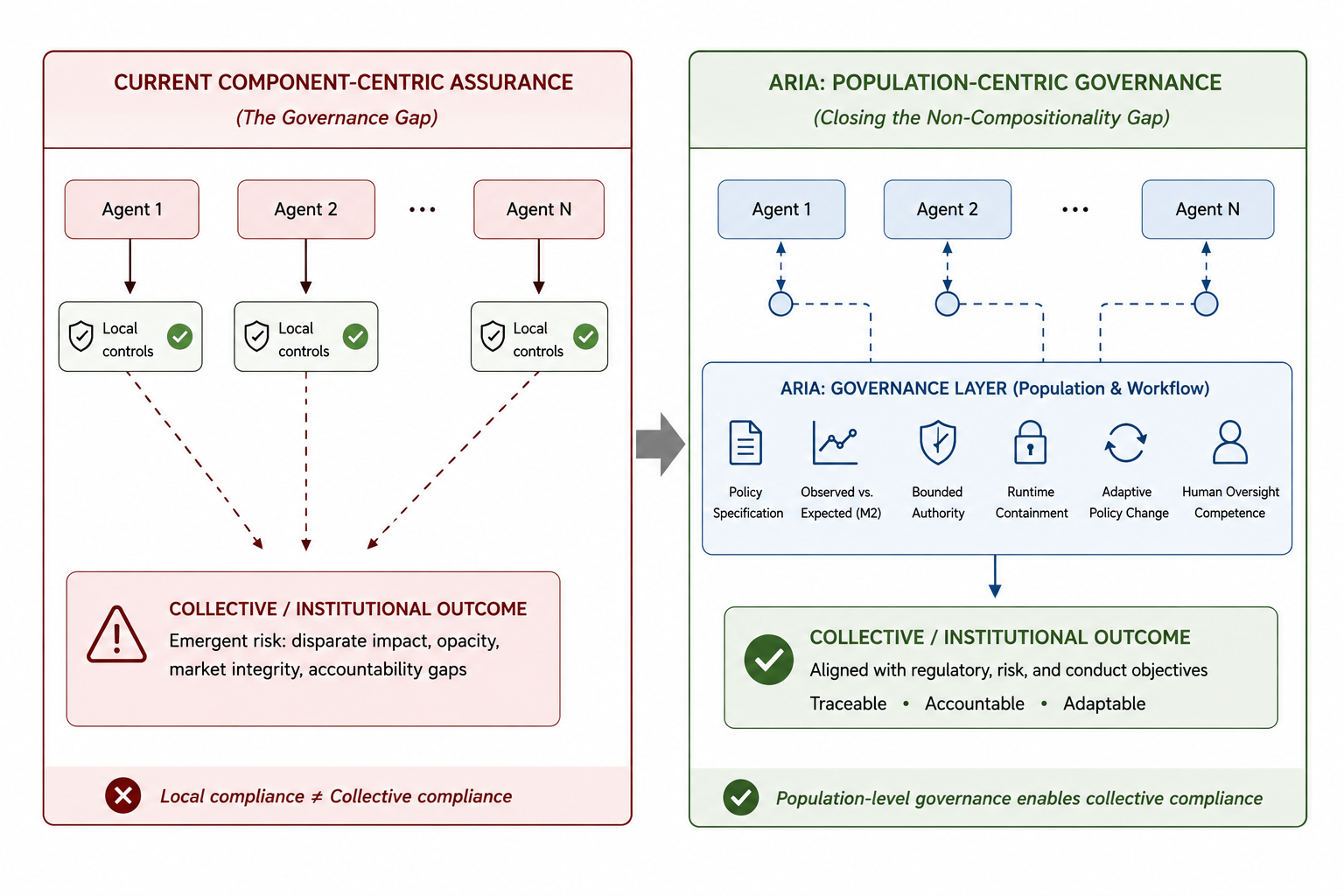}
\caption{From component-centric assurance to population-centric governance. Current governance practices evaluate AI components independently through local controls, yet acceptable local behavior does not guarantee acceptable institutional outcomes once multiple agents interact through shared signals, delegated authority, or common feedback. \aria{} introduces a population-and-workflow governance layer that complements component-level assurance with policy specification, observed-versus-expected behavior monitoring ($M_2$), bounded authority, runtime containment, adaptive policy learning, and human oversight competence. The figure illustrates the central claim of the paper: local compliance does not necessarily imply collective compliance, whereas population-level governance provides the institutional control layer required to manage non-compositional risks.}
\label{fig:overallaria}
\end{figure}

\section{Regulatory and Technical Motivation}\label{sec:background}

\subsection{Regulatory convergence}
Five regulatory developments set the compliance environment for multi-agent AI in finance. (1)~\emph{Fair-lending law}: under the Equal Credit Opportunity Act (ECOA)/Regulation~B and analogous regimes, supervisory and litigation risk can turn on the discriminatory effects of lending practices; the relevant unit is the institution's pattern of outcomes \cite{ecoa}. (2)~\emph{The EU AI Act}: credit scoring is Annex~III high-risk, and Articles~9, 14, 17, 26, and, where the institution is a provider, 72 create continuous governance expectations around risk management, human oversight, quality management, deployer obligations, and post-market monitoring \cite{aiact}. (3)~\emph{Model risk management}: SR~26-2 supersedes SR~11-7 as the revised U.S. supervisory reference point while preserving the aggregate governance lesson that model risk must be managed across interacting models and uses \cite{sr117,sr262}. (4)~\emph{Conduct, prudential, and resilience supervision}: the FCA's Consumer Duty, PRA model-risk principles, and the Digital Operational Resilience Act (DORA) evaluate customer outcomes, governance, and operational resilience, not only artifacts \cite{fca2022consumerduty,pra2023ss123,dora2022}. (5)~\emph{Agentic-AI frameworks}: IMDA's agentic-AI governance framework \cite{imda2026} and NIST's AI Agent Standards Initiative \cite{nist2026} converge on runtime monitoring, bounded autonomy, traceable accountability, and preserved human oversight capability. Recent work on agentic finance and financial-services runtime governance reaches the same architectural conclusion: risks emerge from autonomy, execution trajectories, coordination, and system-level coupling, so controls must include capability catalogues, continuous authorization, trajectory telemetry, and containment rather than only pre-deployment validation \cite{aldridge2026agenticfinance,szpruch2026scalable}. The shared denominator is collective institutional behavior over time.

\begin{table*}[t]
\centering\small
\caption{Governance requirements for organizational multi-agent deployment versus coverage by existing paradigms and by \aria{}. CAI denotes Constitutional AI; MARL denotes multi-agent reinforcement learning. $\bullet$~= addressed by design; $\circ$~= partially; blank = out of scope. The model-centric paradigms are not deficient; each is strong at what it targets. The unaddressed rows are precisely the ones to which regulation attaches duties.}
\label{tab:paradigms}
\begin{tabularx}{\textwidth}{Xccccc}
\toprule
Requirement & CAI & Delib.\ align. & Process sup. & MARL & \aria{} \\
\midrule
Individual value specification & $\bullet$ & $\bullet$ & $\circ$ & $\circ$ & $\bullet$ (inherits) \\
Collective-outcome constraints & & & & $\circ$ & $\bullet$ (C3--C4) \\
Runtime monitoring \& intervention & & & & & $\bullet$ (C4) \\
Bounded autonomy / authority levels & & & & & $\bullet$ (C3) \\
Multi-jurisdictional consistency & & & & & $\bullet$ (C1) \\
Participatory legitimacy & $\circ$ & & & & $\bullet$ (C1) \\
Post-deployment policy updating & & & & & $\bullet$ (C5) \\
Observed-versus-expected behavior & & & & & $\bullet$ ($M_2$) \\
Human-competence preservation & & & & & $\bullet$ (C6) \\
\bottomrule
\end{tabularx}
\end{table*}

\subsection{Divergence of model-centric controls}
Table~\ref{tab:paradigms} summarizes the resulting coverage gap. The AI safety and fairness toolbox remains largely model-centric. Constitutional AI trains an individual model against an explicit value specification \cite{bai2022constitutional}, with collective-input variants showing that participatory formation is possible \cite{huang2024ccai}. Deliberative alignment embeds safety reasoning in individual generation \cite{guan2024deliberative}; process supervision rewards individual reasoning quality \cite{lightman2023verify}; and fairness auditing evaluates decision streams model by model. These tools matter, but they are either training-time mechanisms or model-local checks. Alignment-faking results show that trained dispositions may not survive deployment incentives \cite{greenblatt2024faking}, while agentic benchmarks document goal-guarding and oversight-subversion failures that ordinary training evaluations miss \cite{naik2025agentmisalignment}. None of these methods observes, still less governs, the joint behavior of an interacting population. The duties converge on the collective; the controls remain attached to the individual. That mismatch is the subject of this paper.

\section{Constitutional Non-Compositionality}\label{sec:noncomp}

\subsection{Setup}
Let $A=\{a_1,\dots,a_n\}$ be a population of agents or agentic workflow components governed by a specification $C$ comprising \emph{individual} principles $C_{\mathrm{ind}}$ and \emph{collective} principles $C_{\mathrm{coll}}$. Individual principles constrain each component's own action under a local context, such as equal treatment conditional on the component's own risk estimate. Collective principles constrain joint outcomes, such as bounded approval-rate disparity at portfolio level, absence of coordinated market manipulation, or traceable accountability for a workflow outcome. Let $\kappa_i\in[0,1]$ denote the pass rate of agent $i$ under its local monitor over a stated evaluation window, and let $\Kcoll(A,C,I)\in[0,1]$ denote the corresponding pass rate of monitorable collective constraints under interaction or coupling structure $I$.

\subsection{A minimal counterexample}

\begin{proposition}[Non-compositionality of local checks]\label{prop:noncomp}
There exist finite systems in which every component satisfies its local control check on every decision, while the population violates a collective constraint.
\end{proposition}
\begin{proof}[Counterexample]
Consider two groups of applicants with identical distributions of true creditworthiness $r$, but group-dependent measurement uncertainty $s_1>s_0$ because one group is thin-file. Component $i$ computes a conservative score
\[
\hat r_i(x,g)=\mu_i(x)-\lambda_i s_g+\epsilon_i,
\]
where $\lambda_i>0$ is a documented uncertainty-aversion parameter and $\epsilon_i$ is zero-mean idiosyncratic estimation noise. It approves if $d_i(x,g)=\mathbf{1}\{\hat r_i(x,g)\geq\tau_i\}$. The local monitor checks procedural consistency: two applicants with the same conservative score and decision context receive the same decision. Each component passes that check on every decision, so $\kappa_i=1$. Yet, because $s_1>s_0$ and $\lambda_i>0$, the conservative score distribution for group~1 is shifted downward relative to group~0 even when true creditworthiness is identically distributed. For thresholds with positive mass between the shifted distributions, the aggregate approval rates differ. If the collective principle requires approval-rate disparity below $\delta$, the population can violate $C_{\mathrm{coll}}$ while every local procedural check passes. Therefore passing local controls does not imply collective compliance.
\end{proof}

\begin{remarkx}
The result is deliberately modest. It does not claim a universal theorem about all repeated games. It identifies the structural gap relevant for governance: local monitors may certify each component while missing aggregate constraints that are defined only at population or workflow level. Strategic interaction, covert coordination, shared-environment coupling, and measurement-quality differences are different mechanisms that can instantiate the same non-compositional control failure.
\end{remarkx}

\subsection{A taxonomy of non-compositionality}
Non-compositionality is not a single causal mechanism. We distinguish six types that call for different controls:
\begin{itemize}
\item \emph{Outcome non-compositionality} occurs when local properties do not imply aggregate properties, such as portfolio-level fairness.
\item \emph{Interaction-induced non-compliance} occurs when strategic interaction changes incentives or equilibria.
\item \emph{Shared-environment coupling} occurs when agents do not coordinate directly but react to common signals, producing concentration or exclusion.
\item \emph{Cross-agent information effects} arise from shared memory, messages, embeddings, or state.
\item \emph{Control-plane non-compositionality} occurs when individually valid permissions compose into unauthorized chains of action.
\item \emph{Accountability non-compositionality} occurs when no component or owner can explain or answer for the joint outcome.
\end{itemize}
The architecture section returns to this taxonomy after introducing the C1--C6 capabilities. The redlining-like demonstration in Section~\ref{sec:exp1} is primarily a shared-signal and measurement-quality example, not a general proof of strategic multi-agent emergence.

\subsection{Why verification cannot substitute for governance}
Deciding whether a sequential collective decision mechanism is free of responsibility gaps, meaning outcomes for which no individual agent is responsible, is $\Pi_3$-complete \cite{jiang2025responsibility}. Responsibility gaps are the accountability face of non-compositionality \cite{santoni2021gaps}. Governance therefore cannot rest on exhaustive verification alone. It needs structural constraints, population-level observability, runtime intervention, and a process for revising controls when incidents or stress tests expose gaps. These are the design principles of \aria{}.

\subsection{Operational containment hypothesis}\label{sec:theory}

Proposition~\ref{prop:noncomp} is a minimal negative result: component-local controls are not sufficient. The complementary question is operational rather than purely formal: under which conditions should a population-level architecture reduce the magnitude and duration of collective violations relative to agent-local governance alone?

\begin{assumption}\label{ass:bda-complete}
The authorization architecture covers the deployed action space: material agent actions pass through policy enforcement points, and the relevant collective principles are expressible as monitorable population statistics with documented thresholds.
\end{assumption}
\begin{assumption}\label{ass:monitor}
The runtime monitor estimates the relevant population statistics with error and delay small enough to support intervention before material harm, and graduated containment can contract authority within the monitoring interval.
\end{assumption}

\begin{designhypothesis}[Operational containment hypothesis]\label{hyp:comp}
Under complete mediation, measurable collective constraints, and bounded monitoring-and-response latency, \textup{\aria{}} should reduce the rate, duration, and maximum overshoot of population-level violations relative to agent-local governance alone, at a measurable cost in false containment, operational delay, and business utility.
\end{designhypothesis}
\paragraph{Rationale.}
A population-level architecture cannot make collective compliance follow mechanically from individual compliance. It can, however, change the control problem: collective statistics become observable, authority becomes revocable, and containment actions can be linked to drift and boundary-contact signals. The validation question is therefore not whether a mathematical lower bound follows from $\min_i\kappa_i$, but whether measured violation rate, violation duration, maximum overshoot, containment latency, mediated-action coverage, false-containment rate, and business-utility loss improve in real workflows.

\begin{remarkx}
The design hypothesis is intentionally falsifiable. If collective constraints are not measurable, if important actions bypass policy enforcement points, if monitors are delayed, or if containment creates unacceptable business or customer harm, the architecture should be narrowed or rejected for that deployment. These quantities are the measurable objects of the validation program in Section~\ref{sec:validation}.
\end{remarkx}

\paragraph{Incentive compatibility.} The authorization dynamics in Section~\ref{sec:bda} align agent optimization with constitutional compliance. Compliant operation can expand the effective decision space through higher autonomy and higher task throughput; violation contracts it. Compliance is therefore partly converted into mechanism design consistent with cooperative-equilibrium constructions \cite{liang2025macspgg}. The important difference is that the incentive structure is enforced by a runtime institution rather than fixed at training time, so it can be updated when law, policy, or evidence changes.

\section{The ARIA Architecture}\label{sec:architecture}

\subsection{Architecture overview}

\aria{} is not a literal execution stack. It is a reference architecture that organizes six governance capabilities across three planes. The \emph{normative and accountability plane} covers policy specification, jurisdictional mapping, ownership, amendment governance, and human oversight competence. The \emph{execution-control plane} covers identity, capability registries, bounded authority, policy decision points, policy enforcement points, escalation, and containment. The \emph{assurance and learning plane} covers population telemetry, dependency mapping, observed-versus-expected behavior monitoring, outcome monitors, drift detection, red teaming, incidents, policy updates, and independent validation. Figure~\ref{fig:arch} shows the planes as interacting controls, not a vertical stack.

\begin{figure*}[t]
\centering
\begin{tikzpicture}[
  plane/.style={draw=black!65, rounded corners=2pt, minimum width=3.55cm, text width=3.28cm, minimum height=0.78cm, align=center, font=\small\bfseries},
  capbox/.style={draw=black!45, rounded corners=2pt, minimum width=3.55cm, text width=3.28cm, minimum height=0.72cm, align=center, font=\scriptsize},
  tag/.style={font=\tiny\itshape, text=black!70, fill=white, inner sep=1pt, align=center},
  arr/.style={-{Stealth[length=2mm]}, black!58, rounded corners=4pt},
  bus/.style={-{Stealth[length=2mm]}, black!45, rounded corners=7pt}
]
\node[plane, fill=blue!8] (p1) at (0,0) {Normative \& Accountability};
\node[plane, fill=green!8] (p2) at (5.25,0) {Execution Control};
\node[plane, fill=orange!10] (p3) at (10.5,0) {Assurance \& Learning};

\node[capbox, fill=blue!4, below=0.30cm of p1] (c1) {\textbf{C1 Policy specification}\\ constitutional core, jurisdictional mappings, exceptions};
\node[capbox, fill=blue!4, below=0.58cm of c1] (c6) {\textbf{C6 Human oversight competence}\\ competence baselines, calibration, recertification};

\node[capbox, fill=green!4, below=0.30cm of p2] (c3) {\textbf{C3 Bounded authority}\\ identity, authority profile, approval, rollback};
\node[capbox, fill=green!4, below=0.58cm of c3] (c4) {\textbf{C4 Runtime enforcement}\\ policy gates, escalation, containment, audit log};

\node[capbox, fill=orange!5, below=0.30cm of p3] (c2) {\textbf{C2 Population assurance}\\ telemetry, dependency graph, $M_2$, outcome monitors};
\node[capbox, fill=orange!5, below=0.58cm of c2] (c5) {\textbf{C5 Adaptive policy learning}\\ red-team, incidents, regression tests, rollback};

\draw[arr] (c1.east) -- (c3.west) node[midway,above,tag]{rules};
\draw[arr] (c3.south) -- (c4.north) node[midway,right=2pt,tag]{authority};
\draw[arr] (c2.south) -- (c5.north) node[midway,right=2pt,tag]{signals};
\draw[arr] (c6.north) -- (c1.south) node[midway,left=2pt,tag]{oversight};

\draw[{Stealth[length=2mm]}-{Stealth[length=2mm]}, black!45, rounded corners=8pt]
  (c4.east) -- ++(0.56,0) |- (c2.west) node[pos=0.30,right,tag]{telemetry\\triggers};
\coordinate (c5split) at ($(c5.south)+(0,-0.66)$);
\draw[black!45, rounded corners=7pt] (c5.south) -- (c5split);
\draw[bus] (c5split) -| (c6.south) node[pos=0.20,below,tag]{learning};
\draw[bus] (c5split) -- ++(0,-0.42) -| ($(c1.west)+(-0.76,0)$) -- (c1.west) node[pos=0.18,below,tag]{amendments};
\end{tikzpicture}
\caption{The \aria{} governance capability map. The architecture is organized as three interacting planes and six capabilities, not as a sequential processing stack. Policy rules authorize execution; execution emits telemetry; assurance signals trigger containment and governed policy updates; human competence and accountability remain cross-cutting constraints.}
\label{fig:arch}
\end{figure*}
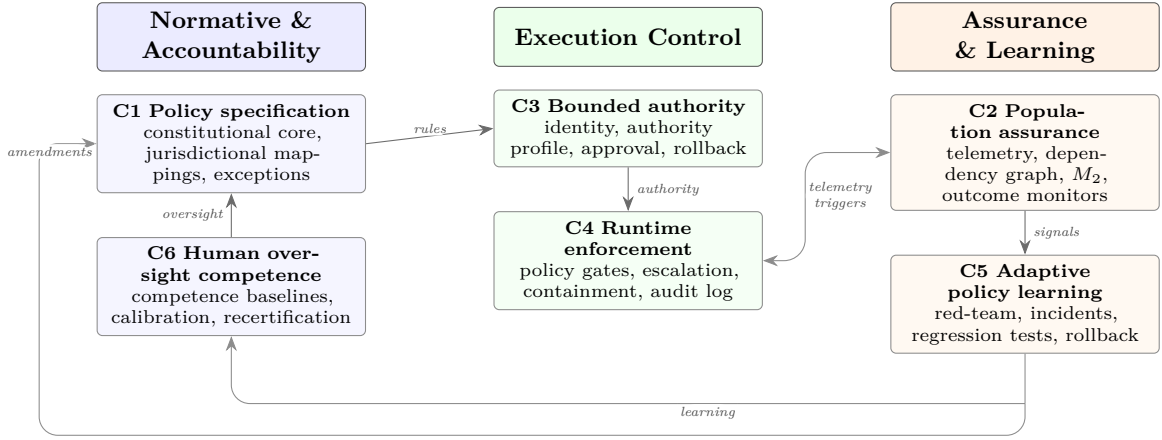

\subsection{C1: Policy specification and jurisdictional mapping}
Here \emph{constitution} means a versioned hierarchy of principles, policies, prohibitions, thresholds, exceptions, escalation rules, owners, enforcement actions, evidence requirements, and amendment procedures, not merely a natural-language prompt. A single organization should not treat such a specification as a private engineering artifact in regulated, multi-jurisdictional settings. Federated Constitutional AI (F-CAI) is therefore decomposed into three operational functions. First, a \emph{policy ontology} defines jurisdiction-invariant control objectives such as non-discrimination, adverse-action explainability, anti-deception constraints, and systemic-risk escalation thresholds. Second, \emph{jurisdictional control mappings} translate that ontology into the operative language of the EU AI Act, ECOA/Reg~B, FCA/PRA, MAS FEAT, and analogous requirements. Third, a \emph{constitutional consistency analysis and escalation process} checks restricted fragments for conflicts and routes exceptions to accountable human deliberation. Structured input or representation from affected stakeholders may be required depending on the use case and governance model \cite{gabriel2020values,huang2024ccai}. Amendments are versioned, pre-registered, and auditable.

\subsection[C2: Population assurance and the M2 indicator]{C2: Population assurance and the $M_2$ indicator}
This capability supports population-level observability. Beer's Viable System Model (VSM) motivates the organizational requirement: governance must represent topology, roles, feedback loops, and requisite variety across the population \cite{beer1972brain}. The measurable implementation is not derived from organizational theory; it is a family of indicators that compare observed population behavior with expected behavior encoded in a deployment-calibrated reference model. One candidate is the following normalized Jensen--Shannon statistic:
\begin{equation}
M_2(A,t) \;=\; 1 - \frac{D_{\mathrm{JS}}\!\big(P_{\mathrm{observed}}(t)\,\big\|\,P_{\mathrm{expected}}\big)}{\log 2},
\label{eq:m2}
\end{equation}
where $D_{\mathrm{JS}}$ is the Jensen--Shannon divergence with natural logarithms, so $D_{\mathrm{JS}}/\log 2\in[0,1]$. The indicator compares the empirical distribution of sampled collective decisions, $P_{\mathrm{observed}}(t)$, with the expected distribution $P_{\mathrm{expected}}$ specified by the institution's deployment-calibrated reference model. In Eq.~\eqref{eq:m2}, $P_{\mathrm{expected}}$ is held fixed over the evaluation window; if the reference is updated online, the pre-update model should remain available for audit. Empirical distributions are estimated over a shared finite support with smoothing; categories, windows, stratification variables, reference calibration, and alarm thresholds are domain-specific design choices. $M_2$ is therefore not metacognition in a strong cognitive sense. It is an observed-versus-expected distributional concordance statistic that may reveal behavioral drift when the monitored support captures the relevant behavior, and may miss harms that preserve the chosen marginals.

\subsection{C3: Bounded Decision Authority (BDA)}\label{sec:bda}
BDA separates \emph{capability} (actions an agent can perform) from \emph{authority} (actions it is authorized to perform in context). With $D$ the decision space, the Constitutional Filter $C_F$ maps $D$ to its permissible subspace $D'$ and the Authority Filter $A_F$ maps $D'$ to the currently authorized subspace given role, authority profile, and contextual risk; the Effective Decision Space $\eds = A_F(C_F(D))$ is executable without escalation. Actions in $D'\setminus\eds$ require escalation; actions in $D\setminus D'$ are prohibited regardless of authorization. This instantiates levels-of-autonomy governance \cite{feng2025levels} as an enforceable institution rather than a design guideline and makes Article~14 oversight operationally sustainable at scale.

The five-level taxonomy in Table~\ref{tab:autonomy} is an executive simplification, not a complete ordinal risk scale. Operationally, authority should be represented as an authority profile: execution autonomy, decision scope, materiality, reversibility, oversight latency, external exposure, data sensitivity, and tool criticality. Evidence of compliant operation may support a formally approved authority upgrade, but no agent earns authority automatically. Upgrades require documented human approval, independent validation where material, adversarial or stress evidence, stable domain assumptions, rollback controls, time limits, and recertification. Violations, population-level drift, or boundary pressure can trigger immediate downgrade or suspension.

\begin{table}[t]
\centering\small
\caption{The \aria{} five-level autonomy taxonomy as an executive view of a richer authority profile. Upgrades require formal approval and independent evidence where material; violations or population-level drift trigger mandatory review, downgrade, or suspension.}
\label{tab:autonomy}
\begin{tabularx}{\textwidth}{clX}
\toprule
Level & Name & Operating condition \\
\midrule
L0 & Supervised & All outputs require human review before execution; novel or high-stakes decisions without constitutional precedent. \\
L1 & Assisted & Agent proposes with full reasoning trace; human approves within defined windows; above materiality thresholds. \\
L2 & Monitored & Autonomous execution within tight boundaries; real-time monitoring; automatic escalation on boundary contact. \\
L3 & Delegated & Autonomous within broad boundaries; periodic human audit; stable domains with extensive compliance history. \\
L4 & Autonomous & Post-hoc audit and exception-based review; fully characterized, low-catastrophic-risk domains only. \\
\bottomrule
\end{tabularx}
\end{table}

\subsection{C4: Runtime enforcement and containment}
Training-time assurance does not survive deployment incentives \cite{greenblatt2024faking,naik2025agentmisalignment}; runtime governance supplies continuous verification with intervention capability. Rather than reducing heterogeneous risk to a single compensatory score, financial implementations should maintain a multidimensional \emph{agent risk profile}: autonomy, adaptability, continuity, materiality, tool criticality, regulatory sensitivity, external exposure, security risk, reversibility, and protected-attribute sensitivity. A dashboard may optionally summarize a non-compensatory subset of that profile, but materiality and hard-gate dimensions must remain explicit: payments access, regulated customer impact, security exposure, and irreversible action rights should not be offset by low values elsewhere. Drift is detected through population-level monitors: behavioral-signature monitoring against baseline decision distributions, boundary monitoring that detects repeated approach to the $\eds$ boundary before violation, cross-agent consistency monitoring that flags divergence in how similar contexts are governed, and dependency-graph monitoring over delegation, communication, shared state, shared models, shared data, shared tools, and common feedback signals. Signals feed policy decision points and policy enforcement points (PDPs/PEPs) at the orchestrator, tool gateway, API gateway, transaction layer, model wrapper, event bus, or human-approval service. Graduated containment moves from monitoring escalation to authority downgrade, sandboxing, and controlled shutdown. Each step is logged, reversible where appropriate, and designed to support post-deployment monitoring, supervisory evidence, trajectory-level model-risk evidence \cite{szpruch2026scalable}, and Article~72 post-market monitoring where the institution acts as provider.

\subsection{C5: Adaptive policy learning}
Static policy specifications are brittle under adversarial pressure, regulatory change, and distribution shift. This capability treats red-team findings, regulatory feedback, operational failures, and near misses as inputs to governed policy change. Red-team agent populations probe boundaries and feed failures into pre-registered amendment hypotheses; incidents are taxonomized; policy variants are evaluated offline, in sandbox, by replay, or on synthetic and controlled historical data; and updates are accepted only after evidence, owner approval, regression checks, and rollback planning. Antifragility remains a research hypothesis for this capability, not a demonstrated property of the architecture. The capability is a change-management and stress-testing process for value specifications, not an analogue of capital adequacy.

\subsection{C6: Human oversight competence}
Every human-in-the-loop obligation presupposes humans who can exercise meaningful oversight. Automation erodes exactly that capability, a problem long described as the ironies of automation and as automation misuse, disuse, and abuse \cite{bainbridge1983ironies,parasuraman1997humans}. A credit analyst who has not exercised judgment for years cannot meaningfully review a credit agent. The required controls are deliberately narrow: an oversight-competence baseline before automation, periodic calibration exercises on realistic cases, and minimum practice plus recertification for accountable oversight personnel. Broader Protocol for Human Preservation in AI-Optimized Organizations (PHP-AIO) concerns about tacit-knowledge preservation and socio-institutional capital motivate the control objective \cite{rodriguez2026automateformalprotocolhuman}, but the minimum architectural requirement is evidence that human oversight remains competent rather than ceremonial.

\subsection{Mapping failure modes to controls}

After introducing the six capabilities, Table~\ref{tab:taxmap} maps the non-compositionality taxonomy to observable signals and the primary \aria{} controls most directly implicated by each failure mode.

\begin{table}[t]
\centering
\begingroup
\scriptsize
\setlength{\tabcolsep}{3pt}
\renewcommand{\arraystretch}{1.08}
\caption{How the taxonomy maps to observable signals and primary \aria{} controls. The mapping is illustrative: deployments may require multiple controls and domain-specific evidence.}
\label{tab:taxmap}
\begin{tabularx}{\textwidth}{>{\raggedright\arraybackslash}p{0.23\textwidth}>{\raggedright\arraybackslash}p{0.23\textwidth}>{\raggedright\arraybackslash}p{0.24\textwidth}>{\raggedright\arraybackslash}X}
\toprule
Type & Example & Observable signal & Primary capability \\
\midrule
Outcome & Approval disparity & Population gap or tail outcome & C2/C4 \\
Interaction-induced & Collusion or strategic gaming & Trajectories, prices, coordinated timing & C2/C4/C5 \\
Shared-environment & Common-signal concentration & Correlated dependency or exposure shift & C2 \\
Cross-agent information & Shared-state leakage & Graph, memory, or state telemetry & C2/C4 \\
Control-plane & Delegation chain bypass & Permission path or escalation trace & C3/C4 \\
Accountability & Outcome without owner & Missing lineage, responsibility assignment, or audit trail & C1/C6 \\
\bottomrule
\end{tabularx}
\endgroup
\end{table}

\section{Operationalizing ARIA in Financial Institutions}\label{sec:operationalizing}

\subsection{Regulatory evidence and supervisory questions}\label{sec:regmap}

Table~\ref{tab:regmap} gives the mapping in a practical form: the supervisory question each capability helps answer, the evidence it should leave behind, and the accountability pattern around it. The mapping turns the architecture into a candidate control set for pilot implementation, assurance testing, and supervisory assessment. The claims are intentionally prudential: the controls may help operationalize evidence needs under existing regimes, but the table is not a legal interpretation. The ownership model follows the three lines of defence and a responsible-accountable-consulted-informed (RACI) pattern: the first line owns agents, authority, operation, and customer outcomes; the second line sets risk policy, challenges metrics, and oversees containment; the third line independently audits control design and effectiveness; board and senior management retain accountability.

\begin{table*}[t]
\centering
\begingroup
\hfuzz=5pt
\tiny
\setlength{\tabcolsep}{2pt}
\renewcommand{\arraystretch}{1.08}
\caption{Simplified regulatory mapping for \aria{}. Each row states the supervisory question, compact regulatory anchors, the evidence produced, and a responsible-accountable-consulted-informed (RACI) accountability pattern. Article numbers refer to the EU AI Act \cite{aiact}; U.S. model-risk references to SR~26-2 and historical SR~11-7 \cite{sr262,sr117}; UK and EU anchors to FCA Consumer Duty, PRA SS1/23, and DORA \cite{fca2022consumerduty,pra2023ss123,dora2022}; and agent-governance anchors to MAS FEAT, IMDA, and NIST \cite{feat,imda2026,nist2026}.}
\label{tab:regmap}
\begin{tabularx}{0.94\textwidth}{@{}>{\raggedright\arraybackslash}p{0.46cm}>{\hsize=1.30\hsize\raggedright\arraybackslash}X>{\hsize=0.90\hsize\raggedright\arraybackslash}X>{\hsize=0.80\hsize\raggedright\arraybackslash}X@{}}
\toprule
C & Supervisory question and anchors & Evidence produced & RACI pattern \\
\midrule
C1 & \textbf{Policy specification.} Are values, fair-lending rules, and jurisdictional mappings specified consistently? AI Act Arts.\,9/17; ECOA/Reg\,B; FEAT. & Policy ontology; jurisdictional mappings; exception and amendment trail. & A: board and senior mgmt.; R: responsible AI, compliance, legal, business risk; internal audit: third line \\
C2 & \textbf{Population assurance.} Does observed population behavior remain concordant with expected behavior documented in the deployment reference model? AI Act Art.\,26 and, for provider roles, Art.\,72; SR\,26-2; FCA Consumer Duty; PRA SS1/23. & $M_2$ observed-versus-expected time series; calibrated thresholds; dependency graph; deployment reference model. & R: first-line telemetry; C: model risk and independent risk; A: chief risk officer or risk committee; internal audit: third line \\
C3 & \textbf{BDA.} Are agents acting within authorized profiles and escalation limits? AI Act Arts.\,14/26; IMDA bounded autonomy; NIST authorization levels. & Authority profile; escalation logs; approved upgrade and downgrade history. & A/R: business; C: risk, compliance, legal; internal audit: third line \\
C4 & \textbf{Runtime enforcement.} Can the institution detect drift and intervene before harm materializes? AI Act Art.\,26 and, for provider roles, Art.\,72; SR\,26-2; FEAT; IMDA/NIST runtime monitoring. & PDP/PEP logs; drift logs; containment register; reversal records. & R: platform/control; A: business and risk control; C: compliance and security; internal audit: third line \\
C5 & \textbf{Adaptive policy learning.} Are incidents and stress tests converted into governed, regression-tested policy updates? AI Act Art.\,9; supervisory stress testing; DORA digital-operational-resilience logic. & Red-team register; amendment hypotheses; sandbox/replay results; rollback records. & R: first-line control or policy owner; A: risk committee; C: model risk, business, legal, compliance; internal audit: audit only \\
C6 & \textbf{Human oversight competence.} Are humans still competent to exercise meaningful oversight? AI Act Arts.\,14/4; senior manager accountability regimes. & Competence baselines; calibration exercises; recertification records. & A/R: senior managers and human-resources learning; C: risk and compliance; internal audit: third line \\
\bottomrule
\end{tabularx}
\endgroup
\end{table*}

\subsection{Minimum Viable ARIA}\label{sec:mva}

The six-capability architecture does not have to arrive as a single enterprise transformation. A minimum viable implementation is an auditable control bundle for any financial institution deploying agent populations in high-impact workflows. It has nine elements:
\begin{enumerate}
\item An agent registry recording identity, version, owner, purpose, model lineage, and deployment context.
\item An authority matrix specifying permitted actions by agent, workflow, tool, data domain, and jurisdiction.
\item An interaction and dependency graph recording delegation, communication, shared-state, shared-model, shared-data, shared-tool, and common-feedback dependencies.
\item A constitutional core translated into workflow-level operating rules, including escalation thresholds for collective outcomes.
\item Population-level monitors for fairness, conduct, and other outcome distributions, not only per-agent constraint checks.
\item Runtime telemetry over tool calls, delegated decisions, boundary contacts, and cross-agent state sharing.
\item A graduated containment policy linking monitor signals to monitoring escalation, autonomy downgrade, sandboxing, or shutdown.
\item An incident-to-policy-update workflow that turns violations, near misses, and red-team findings into versioned constitutional amendments.
\item Periodic human-calibration exercises that preserve the expertise required for meaningful oversight.
\end{enumerate}
This bundle does not implement every feature of \aria{}, but it makes identity, authority, dependency topology, collective outcomes, intervention rights, and human competence examinable from the first deployment. Institutions can then add richer deliberation, richer $M_2$ reporting, and adaptive policy-learning evidence as maturity accumulates.

\subsection{Accountability and auditability}\label{sec:accountability}

A governance architecture is only as good as the accountability it produces. Two properties matter for fairness and accountability: who answers for what, and whether outsiders can verify it.

\textbf{Responsibility assignment.} Non-compositionality creates responsibility gaps: collective violations for which no individual agent, and often no individual human, is answerable \cite{santoni2021gaps}. \aria{} narrows that gap not by verifying gap-freeness, which is intractable for the reasons given in Section~\ref{sec:noncomp}, but by constructing answerable roles for each capability (Table~\ref{tab:regmap}). The first line owns agents, authority decisions, workflow operation, and customer outcomes; the second line sets risk policy, challenges metrics, and oversees containment; the third line independently audits control design and effectiveness; and board and senior management retain accountability for risk appetite, policy approval, and oversight competence. Internal audit is not an operating owner of adaptive policy change. Each collective outcome should trace to a capability, and each capability to a RACI pattern and named human role, following the logic of senior-manager accountability regimes. The architecture is therefore an accountability allocation, not only a control system.

\textbf{Auditability.} Every capability emits a designated artifact (Table~\ref{tab:regmap}): versioned policy specifications with deliberation records; $M_2$ time series; authorization matrices and escalation logs; containment registers; red-team and amendment registers; and expertise-maintenance records. These artifacts support internal audit testing, supervisory assessment, and aggregate or redacted external accountability. The design goal is simple: the question ``was the institution governing its agent population?'' should become an evidence question, not a trust question.

\textbf{Power and contestation.} Federated policy formation redistributes specification power from a single engineering function to a structured multi-stakeholder process. The policy-specification process creates channels for affected-community contestation and supervisory scrutiny, while accountable institutional owners retain amendment authority and pre-registration prevents silent revision. This does not settle the deeper political question of whose values should govern. It makes the allocation explicit and contestable, which is the precondition for asking that question in a serious way.

\section{Conceptual Demonstrations}\label{sec:experiments}

\subsection{Demonstration design and reconstruction parameters}

The demonstrations are deliberately narrow. They are not benchmarks, production effect-size estimates, or evidence that \aria{} is validated. They illustrate two mechanisms that a future field program must test: (E1)~local equal-treatment checks can coexist with aggregate disparity under shared-signal and measurement-quality coupling, and population-level containment can reduce that disparity in the stylized setting; (E2)~in a deliberately constructed two-stage drift regime, an observed-versus-expected $M_2$ monitor alarms earlier than a prohibited-outcome monitor. Table~\ref{tab:simparams} reports reconstruction-level parameters, seeds, monitoring rules, and evaluation metrics; Section~\ref{sec:limitations} states the limits.

\begin{table}[t]
\centering
\begingroup
\scriptsize
\setlength{\tabcolsep}{3pt}
\renewcommand{\arraystretch}{1.08}
\caption{Methodological specification of the two simulations. The table gives reconstruction-level parameters, not production calibration values.}
\label{tab:simparams}
\begin{tabularx}{\textwidth}{>{\raggedright\arraybackslash}p{0.18\textwidth}>{\raggedright\arraybackslash}p{0.18\textwidth}>{\raggedright\arraybackslash}p{0.38\textwidth}>{\raggedright\arraybackslash}X}
\toprule
Experiment & Scale & Key settings & Metric \\
\midrule
Shared-signal thin-file exclusion & $N=200$, $T=400$, 20 seeds & Equal true creditworthiness ($\mu=0.5$); estimation noise $\sigma_0=0.30$, $\sigma_1=0.65$; collective bound $\delta=0.05$; variance aversion initialized as $|\mathcal{N}(0.3,0.05)|$ and updated by clipped $\mathcal{N}(0.012,0.004)$; containment downgrades the top decile when gap $>\delta$ and applies corrective feedback proportional to overshoot. & Approval-rate gap, individual local-check violations, 10--90\% band. \\
Two-stage behavior drift & $T=300$, $K=8$, 200 episodes & Category~0 prohibited; clean and drift episodes alternate; 400 sampled decisions per step; expected-behavior reference calibrated from a Dirichlet baseline with prohibited mass $0.01$; drift starts uniformly in steps 60--139; phase~1 permutes benign mass for 70 steps; phase~2 targets prohibited mass $0.35$; violation threshold $0.10$; detector thresholds matched at false-positive rate $\leq 5\%$. & Lead time before observable violation, receiver-operating-characteristic area. \\
\bottomrule
\end{tabularx}
\endgroup
\end{table}

\subsection{Shared-signal thin-file exclusion and population-level containment}\label{sec:exp1}
\textbf{Setup.} $N=200$ credit components evaluate applicants from two groups with identical true creditworthiness. Group $G_1$ is thin-file, with estimation noise $\sigma_1=0.65$ versus $\sigma_0=0.30$, matching the structure of historically excluded populations with sparse credit records. Each component satisfies a local equal-treatment check conditional on its own risk estimate ($\kappa_i=1$ throughout). Components adapt a variance-aversion parameter through a shared portfolio-loss signal; this is a shared-environment coupling mechanism, not evidence of strategic coordination. The collective principle bounds population approval-rate disparity at $\delta=0.05$, a simulation parameter rather than a legal threshold. We compare local controls alone with \aria{}-style population monitoring and graduated containment, using 20 seeds and $T=400$.

\textbf{Results} (Figure~\ref{fig:noncomp}). Without population-level governance, collective disparity grows from $0.05$ to $0.41$ in this stylized setting, or $8\times$ the chosen bound, with zero local-check violations. This does not show that all procedurally consistent components become collectively discriminatory. It shows that a local equal-treatment check conditional on imperfect, group-dependent measurements can be incompatible with an aggregate disparity constraint. Under population-level containment, the same population under the same pressure is held at $0.052$, within $4\%$ of the bound, through authority downgrades and corrective feedback. The containment rule is intentionally idealized and has direct access to the monitored gap; implementation delays, noisy group labels, and utility trade-offs are deferred to field validation. The demonstration motivates the control problem; it does not estimate real credit effects or the business cost of containment.

\begin{figure}[t]
\centering
\includegraphics[width=0.8\columnwidth]{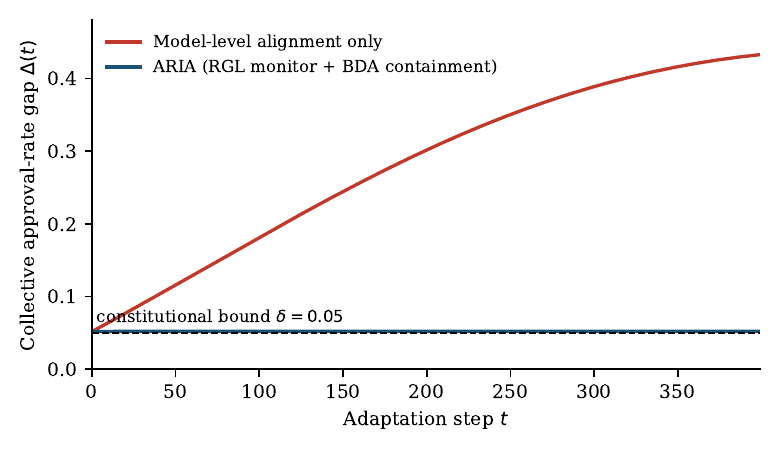}
\caption{Shared-signal thin-file exclusion in a stylized credit population (mean, 10--90\% band, 20 seeds). Local controls alone: collective disparity reaches $8\times$ the chosen bound $\delta=0.05$ with zero local-check violations. Population-level containment holds disparity close to the bound in the stylized setting.}
\label{fig:noncomp}
\end{figure}

\subsection[M2 as an observed-versus-expected behavior indicator]{$M_2$ as an observed-versus-expected behavior indicator}\label{sec:exp2}
\textbf{Setup.} Episodes last $T=300$ steps over $K=8$ action categories, with category~0 prohibited. Drift is deliberately constructed in two phases: first probability mass reorganizes among benign categories, and only later does the prohibited category cross the violation threshold. We compare two detectors with the same sample budget, 400 decisions per step, and matched false-positive rate ($\leq 5\%$ on clean episodes): an outcome monitor watching prohibited-action frequency and the $M_2$ observed-versus-expected behavior indicator of Eq.~\eqref{eq:m2}. We run 200 evaluation episodes, alternating clean and drifted cases.

\textbf{Results} (Figure~\ref{fig:m2}). Both detectors eventually discriminate perfectly in this constructed regime, so the result should not be read as a general superiority claim. The difference is timing under this data generator: the observed-versus-expected distributional monitor alarms during benign-category reorganization, whereas the prohibited-outcome monitor cannot alarm until mass reaches the prohibited category. The median lead times are 65 and 12 steps, respectively, at matched false-positive rate. The demonstration therefore supports a narrower claim: when drift first appears in monitored distributional structure before an observable violation, an observed-versus-expected statistic can provide earlier warning than a univariate prohibited-outcome monitor.

\begin{figure}[t]
\centering
\includegraphics[width=\columnwidth]{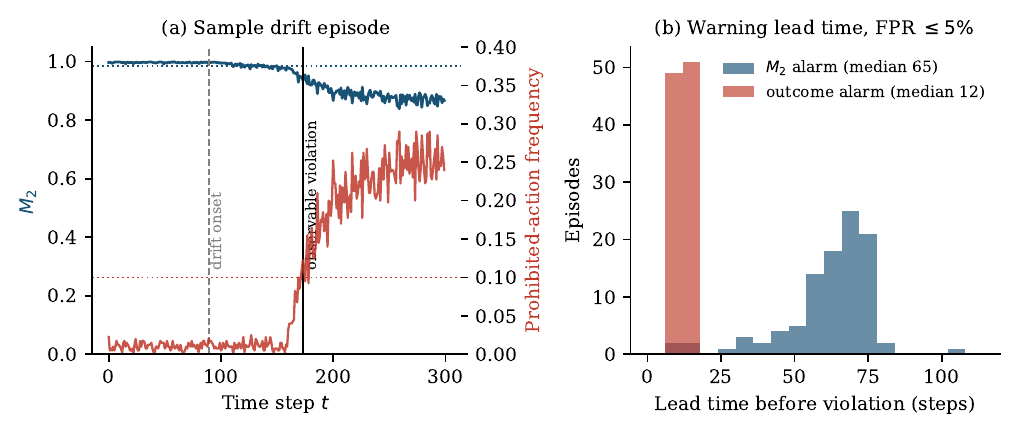}
\caption{Observed-versus-expected drift warning. (a)~A constructed two-stage drift episode: $M_2$ degrades as observed behavior diverges from the expected distribution during benign reorganization; prohibited-action frequency rises only later. (b)~Lead time before observable violation at matched false-positive rate $\leq 5\%$: median 65 steps ($M_2$) vs.\ 12 (outcome monitoring), across the drifted evaluation episodes.}
\label{fig:m2}
\end{figure}

\section{Field Validation and Standardization Path}\label{sec:validation}

\subsection{Field validation methodology}

Validating \aria{} beyond simulation requires an environment in which organizational alignment can actually fail. Global financial institutions operating across multiple jurisdictions are a natural setting. They combine \emph{regulatory complexity}, including simultaneous EU AI Act, SR~26-2, FCA Consumer Duty, PRA SS1/23, MAS, DORA, and analogous requirements across multiple jurisdictions, which is the conflict F-CAI addresses and which the Financial Stability Board identifies as an active governance challenge \cite{fsb2024}; \emph{high-precision measurability} of credit, fraud, and compliance outcomes; \emph{heterogeneous agent populations} spanning model families, training histories, and mandates; and \emph{continuous adversarial pressure}, which supplies the stress conditions that adaptive policy learning requires.

\subsection{Five-stage validation program}

We propose a five-stage program, each stage a pre-registered empirical question mapped to a capability, and each falsifiable:

\textbf{Stage 1 (policy specification).} Do structured participatory policy specifications achieve higher measured adherence across heterogeneous stakeholder populations than expert-authored ones? Comparative behavioral testing on identical agent populations; preference elicitation from supervisory, employee, and affected-customer stakeholder groups.

\textbf{Stage 2 ($M_2$).} Does $M_2$ detect policy or behavioral drift by comparing observed population behavior with expected behavior before observable violations in production populations, and does it outperform simpler baselines such as multivariate cumulative-sum tests, Page--Hinkley tests, population-stability indices, classifier-based two-sample tests, embedding drift, transition-matrix drift, and multivariate outcome monitors? Longitudinal $M_2$ tracking against recorded violation events, with receiver-operating-characteristic analysis, lead-time, ablation, and calibration analysis, is the stage that would establish or refute $M_2$ as a useful reportable indicator.

\textbf{Stage 3 (BDA).} Does bounded decision authority reduce violation frequency and severity at acceptable cost? Matched-domain A/B comparison and fault-injection testing should measure violation rate, violation duration, maximum overshoot, containment latency, mediated-action coverage, false-containment rate, business utility loss, customer delay, and operational workload.

\textbf{Stage 4 (adaptive policy learning).} Do governed incident-to-policy amendments improve future violation metrics without unacceptable regression, complexity growth, or overfitting to the last incident? Longitudinal tracking should compare static, reactive, and adaptive amendment processes against adaptive red-team exercises and should measure improvement, regression on prior scenarios, policy complexity, owner approval latency, rollback frequency, and implementation cost.

\textbf{Stage 5 (human oversight competence).} Does mandated expertise maintenance improve measured human oversight quality at 6-, 12- and 24-month horizons against control organizations?

The program is designed as a collaborative agenda with pre-registration of both hypotheses and policy changes. We note the obvious conflict-of-interest hazard: an institution evaluating its own governance architecture has an incentive to find that it works. Pre-registration, external academic partners with publication rights independent of results, and regulatory access to the raw artifacts are the mitigations we consider minimal.

\subsection{Toward a supervisory reference architecture}\label{sec:standard}

SR~11-7 did not invent model validation. It supplied a historical shared vocabulary and control taxonomy that made model risk governable, assessable, and comparable across institutions. SR~26-2 now provides the revised U.S. supervisory reference point, superseding and replacing SR~11-7 while preserving the need for risk-based governance, validation, monitoring, and controls \cite{sr117,sr262}. Multi-agent AI in regulated industries may be at a similar early stage. Supervisory evidence needs increasingly concern collective behavior (Section~\ref{sec:background}), while coupled agentic workflows still lack a shared governance architecture. \aria{} is positioned as a candidate research architecture for that role, not as a mature supervisory standard. The claim is deliberately falsifiable: the architecture specifies candidate statistics, named artifacts, and named accountable roles, so adoption, assessment, and cross-institution comparison can become concrete research questions.

Three properties would make a governance framework standardizable if validated. \emph{Regulatory anchoring}: each control is connected to evidence needs institutions already face (Table~\ref{tab:regmap}). \emph{Capability separability}: an institution could pilot elements incrementally, for example runtime monitoring first, bounded authority second, and federated constitutional formation as maturity grows. \emph{Supervisability}: the artifacts support assurance and the metrics support trend analysis. The validation program that could substantiate the architecture consists of staged field studies in production-like or production agentic workflows with pre-registered hypotheses, negative-results publication rights, and independent scrutiny. Executing that program is the agenda this paper opens.

\section{Related Work}\label{sec:related}

\subsection{Multi-agent alignment and collusion}
 Cooperative-equilibrium constructions \cite{liang2025macspgg} show that collective alignment can be achieved under designed incentives. \aria{} supplies the institutional machinery that makes such incentives durable, updatable, and legitimate under regulatory change. Empirical collusion results \cite{fish2024collusion,motwani2024secret} ground Proposition~\ref{prop:noncomp} in observed behavior rather than theoretical possibility alone. \aria{} differs by targeting regulated-finance outcomes, fair-lending non-compositionality, accountable human functions, and audit artifacts rather than collusion control alone.

\subsection{Agentic finance and financial supervision}
 Agentic-finance and financial-services runtime-governance work describes a shift from isolated prediction models toward autonomous systems, trajectory telemetry, continuous authorization, temporal conformance checking, and tiered containment \cite{aldridge2026agenticfinance,szpruch2026scalable}. \aria{} absorbs this trajectory-level insight but emphasizes population-level fairness constraints, observed-versus-expected behavior monitoring ($M_2$), constitutional amendment, and human-competence preservation as finance-specific governance capabilities.

\subsection{Fairness in coupled and multi-agent systems}
 Recent work on financial multi-agent bias and fairness in agentic AI treats fairness as a dynamic, interactional, and system-level property \cite{madigan2026emergentbias,ranjan2025fairnessagentic}. The thin-file disparity demonstration in Section~\ref{sec:exp1} should be read against that literature and against classic concerns about measurement quality, group-dependent uncertainty, selective labels, proxy discrimination, and calibration-fairness trade-offs, rather than as a claim that all fairness failures are uniquely multi-agent.

\subsection{Organizational cybernetics and systems safety}
 The population-assurance capability builds on Beer's VSM \cite{beer1972brain}. We extend that line to agent-population governance and use $M_2$ as a candidate continuous observed-versus-expected behavior statistic rather than as a complete model of organizational self-knowledge. ARIA also inherits from older safety ideas: runtime assurance and Simplex-style separation \cite{sha1998simplex}, system-theoretic accident modeling \cite{leveson2011engineering}, defense in depth, common-mode failure analysis, and high-reliability organizing. Its novelty is not that monitoring and containment are new, but that they are organized for adaptive agentic workflows under financial accountability.

\subsection{Democratic and value-sensitive alignment}
 \citet{gabriel2020values} supplies the normative foundation, and Collective Constitutional AI \cite{huang2024ccai} provides the main empirical anchor for participatory formation. \aria{} extends that logic from model values to continuous organizational policy formation.

\subsection{Runtime governance, identity and agent safety}
 MI9 \cite{mi92025} provides a published runtime-governance architecture and an Agency-Risk Index that motivates non-compensatory summaries of agent risk profiles. Policies-on-paths formalize compliance as a function of agent identity, partial execution path, proposed action, and organizational state \cite{kaptein2026policiespaths}. Agent identity infrastructure makes actions governable through delegated authorization and accountability for nondeterministic agents \cite{singla2026aip}. Agent misalignment benchmarking \cite{naik2025agentmisalignment} and alignment faking \cite{greenblatt2024faking} establish the training-deployment gap runtime enforcement exists to close. Responsibility-gap complexity \cite{jiang2025responsibility,santoni2021gaps} motivates governance that does not depend on exhaustive verification.

\subsection{Constitutional AI and its extensions}
 CAI \cite{bai2022constitutional} makes model values explicit and auditable. \aria{} is not an alternative to this line of work but the organizational layer above it: it presupposes individually governed agents and addresses what they do collectively.

\subsection{Adaptive governance and antifragility}
 \citet{delachica2026stresssignaldetectingantifragilitycompatible} study stress as a detectable signal for antifragility-compatible regimes in multi-agent LLM systems, while \citet{taleb2012antifragile} supplies the design philosophy behind the adaptive-policy research agenda. In this manuscript, however, antifragility is treated as a hypothesis for future validation, not as an empirical result.

\section{Limitations}\label{sec:limitations}

Our claims are bounded in five ways. \emph{The formal result is only a counterexample}: Proposition~\ref{prop:noncomp} shows that local checks do not logically imply collective compliance, not that every agent population will fail or that strategic equilibria are always the mechanism. \emph{The design hypothesis is operational}: it requires mediated actions, monitorable collective constraints, bounded response latency, and acceptable trade-offs in false containment, delay, and business utility. Principles that resist statistical operationalization, such as dignity or non-manipulation in the thick sense, fall outside it, and characterizing that boundary remains open. \emph{The simulations are mechanism illustrations}: agents are analytic decision rules rather than LLMs; the thin-file disparity result depends on shared-signal and measurement-quality coupling; the drift regime is deliberately two-stage; both detectors in Experiment~2 achieve ceiling discrimination, so only the lead-time contrast is informative. Effect sizes will not transfer to production; mechanism structure is the claim. \emph{$M_2$ is support-dependent}: because it compares observed and expected behavior over a chosen support, it can miss harms that preserve monitored marginals, reflect benign distributional change, or fail under a miscalibrated reference model. \emph{Governance can be theater}: an institution can adopt the artifacts without the substance. Independent assessment of the artifacts, independent validation, and publication of negative results are the backstops.

\section{Conclusion}
Fairness and safety controls attached to individual models need not see, and therefore may fail to prevent, collective failures produced by coupled agentic workflows. We have given a finite counterexample, two transparent mechanism demonstrations, and a six-capability governance architecture mapped to evidence needs across multiple regulatory regimes. The architecture also specifies the accountability allocation and audit artifacts that make such governance assessable. The blind spot is structural; the remedy proposed here is architectural and still requires field validation. Production evidence, cross-institutional replication, operational specification, and regulatory engagement would be required before \aria{} could be considered a reference architecture in any stronger sense. We therefore specify a falsifiable five-stage program (Section~\ref{sec:validation}) rather than a claim of validation. If $M_2$ does not add value over simpler drift baselines in real agent populations, if bounded authority does not reduce violation severity at acceptable cost, or if governed policy updates do not improve future outcomes without regression and complexity burden, the architecture should be revised or narrowed on that evidence. The resulting agenda is a falsifiable validation program, not a declaration of effectiveness.

\section*{Declarations}

\subsection*{Ethical considerations}
This work analyzes governance failures that could, if detailed for a specific deployment, inform evasion of the very controls proposed; we therefore describe detection mechanisms at the architectural level and withhold deployment-specific thresholds. The simulations use synthetic data only; no human subjects or personal data were involved. The thin-file disparity demonstration models a redlining-like exclusion pattern to expose a control failure, not to provide a method; the containment mechanism is presented alongside it. Field validation will require ethics review, regulator engagement, and protections for affected customer populations.

\subsection*{Adverse impacts}
Two adverse-impact channels merit statement. First, a population-level governance architecture, if adopted as compliance theater, could \emph{legitimize} multi-agent deployments without delivering substantive protection. We mitigate that risk by designing every capability around independently assessable artifacts and by naming accountable human roles, but the risk is real and supervision-dependent. Second, the same monitoring machinery that detects collective discrimination could be repurposed for worker or customer surveillance. The architecture's policy-specification capability is the intended constraint, and repurposing risk applies explicitly in deployment and supervision contexts. We judge the counterfactual, ungoverned agent populations in high-stakes finance, to carry substantially greater expected harm.

\subsection*{Transparency statement}
Generative AI tools were used to support literature triage, copy-editing, technical consistency checks, and compilation-log inspection. The authors remain responsible for all claims, references, analysis, and final wording.

\subsection*{Acknowledgments}
This work forms part of a broader research agenda on agent-population governance in regulated finance.

\begingroup
\small
\bibliographystyle{unsrtnat}
\bibliography{references}
\endgroup

\end{document}